\documentclass[11pt,reqno,oneside, a4paper]{amsart}
\usepackage[ansinew]{inputenc}
\usepackage[english]{babel}
\usepackage[T1]{fontenc} 
\usepackage{lmodern}
\usepackage{amsfonts}
\usepackage{amsmath}
\usepackage{amssymb}
\usepackage{bbm}
\usepackage{amsthm}
\usepackage{arydshln}
\usepackage{xcolor}
\usepackage{enumitem}
\usepackage{mathabx}
\usepackage{comment}
\usepackage{mathrsfs}
\usepackage{url}
\usepackage{hyperref}

\usepackage{cite}
\usepackage{graphicx}
\usepackage{caption}
\usepackage{subcaption}

\usepackage[capitalize]{cleveref}
\crefname{enumi}{}{}
\crefname{equation}{}{}

\usepackage{geometry}

 \specialcomment{supplementary}
{\colorlet{savedcolor}{.} \color{blue} \begingroup \ttfamily \bigskip \smallskip  \noindent \underline{Supplementary details:} \newline \newline \footnotesize }{\endgroup \smallskip \color{savedcolor}}

 \excludecomment{supplementary}

\setenumerate[0]{label=\textnormal{(\roman*)}}

\newcommand{\R}{\mathbb{R}}

\numberwithin{equation}{section}

\newcommand{\E}{\mathbb{E}}
\newcommand{\PR}{\mathbb{P}}

\DeclareMathOperator{\vvec}{vec}
\DeclareMathOperator{\Cov}{Cov}

\newtheorem{thm}{Theorem}[section]
\newtheorem{definition}[thm]{Definition}
\newtheorem{lemma}[thm]{Lemma}

\newtheorem{ass}[thm]{Assumption}
\theoremstyle{remark}

\newtheorem{remark}[thm]{Remark}

\begin{document}

\title[Data driven fitting of multivariate Vasicek models]{Data driven modeling of multiple interest rates with generalized Vasicek-type models}

\author[Ilmonen]{Pauliina Ilmonen}
\address{Aalto University School of Science, Department of Mathematics and Systems Analysis, PO Box 11100, 00076 Aalto, Finland}
\email{pauliina.ilmonen@aalto.fi}

\author[Laurikkala]{Milla Laurikkala}
\address{Aalto University School of Science, Department of Mathematics and Systems Analysis, PO Box 11100, 00076 Aalto, Finland}
\email{milla.laurikkala@aalto.fi}

\author[Ralchenko]{Kostiantyn Ralchenko}
\address{University of Vaasa, School of Technology and Innovations,
P.O. Box 700, Vaasa, FIN-65101, Finland}
\address{Taras Shevchenko National University of Kyiv, Department of Probability, Statistics and Actuarial Mathematics, Volodymyrska St., Kyiv,
01601, Ukraine}

\email{kostiantyn.ralchenko@uwasa.fi}

\author[Sottinen]{Tommi Sottinen}
\address{University of Vaasa, School of Technology and Innovations,
P.O. Box 700, Vaasa, FIN-65101, Finland}

\email{tommi.sottinen@iki.fi}

\author[Viitasaari]{Lauri Viitasaari}
\address{Aalto University School of Business, Department of Information and Service Management, PO Box 11110, 00076 Aalto, Finland}
\email{lauri.viitasaari@aalto.fi}

\keywords{stochastic interest rate models, Vasicek model, parameter estimation, consistency, limiting distribution, Ornstein-Uhlenbeck process}
\subjclass[2020]{91G30,
62M09,
60G10}
\date{September 2, 2025}

\thanks{
P. Ilmonen and M. Laurikkala gratefully acknowledge support from the Research Council of Finland via Finnish Centre of Excellence in Randomness
and Structures 
(decision number~346308). M. Laurikkala gratefully acknowledges support from Emil Aaltonen Foundation (grant number 250110 K1). K. Ralchenko gratefully acknowledges support from the Research Council of Finland (decision number 367468)}

\begin{abstract}
The Vasicek model is a commonly used interest rate model, and there exist many extensions and generalizations of it.
However, most generalizations of the model are either univariate or assume the noise process to be Gaussian, or both.
In this article, we study a generalized multivariate Vasicek model that allows simultaneous modeling of multiple interest rates while making minimal assumptions.
In the model, we only assume that the noise process has stationary increments with a suitably decaying autocovariance structure.
We provide estimators for the unknown parameters and prove their consistencies.
We also derive limiting distributions for each estimator and provide theoretical examples.
Furthermore, the model is tested empirically with both simulated data and real data.
\end{abstract}

\allowdisplaybreaks

\maketitle

{\footnotesize\tableofcontents}

\section{Introduction}
Interest rates are important for economics as they affect the consumption of households and businesses.
Moreover, they can be used as economic indicators since their values often change in accordance with economic shifts.
However, the changes in interest rates are influenced by multiple different factors, and thus interest rates themselves are considered to be stochastic and can be modeled as random processes.
The first model utilizing this idea was the one-factor short-rate interest model introduced by both Merton \cite{merton1973theory} and Black and Scholes \cite{black1973pricing}.
In this model, the interest rate is described using a stochastic differential equation featuring a Brownian motion with drift.
Building on this model, Vasicek \cite{vasicek1977equilibrium} modeled interest rates as an Ornstein-Uhlenbeck process and included mean reversion in the model.

Inspired by \cite{black1973pricing, merton1973theory} and \cite{vasicek1977equilibrium}, many extensions and different versions of interest rate models have been proposed in the literature.
Langetieg \cite{langetieg1980multivariate} proposed a multivariate version of the Vasicek model with a vector of independent Brownian motions as the noise process.
Cox, Ingersoll and Ross (CIR) \cite{cox1985theory} used a similar form to Vasicek but defined the interest rate as a square-root process, disallowing negative values.
The CKLS model \cite{chan1992empirical} combines all these models into one, in which the parameters can be chosen to suit the data at hand.
Hull and White \cite{hull1990pricing} extended the Vasicek and CIR models to allow time-dependent parameters.
This model is sometimes referred to as an extended Vasicek model. 

Most generalizations of the Vasicek model are univariate.
For example, the univariate fractional Vasicek model with univariate fractional Brownian motion as the noise process was introduced already in 1998 by Comte \cite{comte1998long} and its statistical properties were studied in \cite{xiao2019frac1} and \cite{xiao2019frac2}.
Nourdin and Diu Tran \cite{nourdin2019statistical} generalized the univariate fractional Vasicek model further by using Hermite processes as noise.

In this article, we study the generalized multivariate Vasicek model
\begin{equation}
\label{eq:model_intro}
dr_t = \Theta(b-r_t)dt + \sigma dX_t,
\end{equation}
where $X$ is a random process with stationary increments and $r$ is a vector of interest rates.
This model is inspired by the vector-valued generalized Ornstein-Uhlenbeck process studied in \cite{Marko1}, which can be obtained from \eqref{eq:model_intro} by setting $b=0$ and $r_0 = \int_{-\infty}^0 e^{\Theta s}dX_s$.
The generalized multivariate Vasicek model provides a framework to simultaneously model multiple interest rates.
The proposed model is very flexible.
As shown in \cite{Marko1}, every stationary process with continuous sample paths can be viewed as a vector-valued generalized Ornstein-Uhlenbeck process. 
As such, the generalized Vasicek model \eqref{eq:model_intro} covers all interest rates that are (almost) stationary around a mean level $b$.

The rest of the article is organized as follows.
In Section \ref{sec:main}, we introduce the model and present our main results. In particular, we derive estimators for parameters and show that one obtains consistency and asymptotic distributions under rather minimal assumptions.
We illustrate the practical applicability of our approach by presenting examples with simulated data in Section \ref{sec:sim_ex} and a real data example in Section \ref{sec:real_ex}.
Technical proofs are presented in Section \ref{sec:proof}.

\section{Main results}
\label{sec:main}
We consider the generalized multivariate Vasicek model for interest rates $r^{\Theta, b} = \left(r^{\Theta, b}_1,\ldots,r^{\Theta, b}_d\right)$ given by 
\begin{equation}
\label{eq:model}
dr^{\Theta,b}_t = \Theta(b-r^{\Theta,b}_t)dt + \sigma dX_t.
\end{equation}
The initial value $r_0\in \R^d$ is assumed to be a square-integrable random vector. 
The vector $b\in \R^d$ represents the long-term means of the interest rates and $\Theta \in \R^{d\times d}$ is a positive definite matrix modeling the mean-reversion effect and interactions between the interest rates.
The vector $X =(X_t)_{t\geq 0}\in \R^d$ is a square-integrable centered $d$-dimensional noise process with stationary increments with the vector $X_0=0$.
The matrix $\sigma \in \R^{d\times d}$ is a positive definite diagonal matrix describing the size of the noise.

\begin{remark}
\label{rmk:diagonal}
    Without loss of generality, we can assume that the vector $X$ in \eqref{eq:model} has uncorrelated components as for any $X$, we can write $\sigma X = A^{-1}\Tilde{X}$ with $A \in \R^{d\times d}$, where $\Tilde{X}$ has uncorrelated components.
    An invertible matrix $A$ always exists as it can be chosen as $A = W$ with $W$ taken from the eigendecomposition of the cross-covariance matrix $\Sigma = WDW^\top$ of $\sigma X$, where zero eigenvalues have been cut off.
    The matrix $W$ is unitary so $A^{-1} = W^\top = A^\top$.
    Now, multiplying both sides of \eqref{eq:model} by $A$ from the left, we obtain
    \begin{align*}
        d\left(Ar^{\Theta,b}_t\right) = A\Theta A^\top A(b-r^{\Theta,b}_t)dt + d\Tilde{X_t},
    \end{align*}
    which implies
    \begin{equation*}
        d\Tilde{r}^{\Theta,b}_t = \Tilde{\Theta} (\Tilde{b}-\Tilde{r}^{\Theta,b}_t)dt + d\Tilde{X_t},
    \end{equation*}
    where $\Tilde{r}^{\Theta,b}_t = Ar^{\Theta,b}_t$, $\Tilde{b}=Ab$ and  $\Tilde{\Theta} = A\Theta A^\top$.
    The new matrix $\Tilde{\Theta}$ is positive definite, since $\Theta = LL^\top$ is positive definite and $\Tilde{\Theta} = A L L^\top A^\top = (AL) (AL)^\top$. 
\end{remark}
In the sequel, we assume that the vector $X$ has uncorrelated components and finite fourth moments.
Moreover, we use the notation $\Vert \cdot \Vert$ for the $L^2$-norm.
In addition, we drop the superscript from $r^{\Theta, b}$ and simply write $r$ whenever confusion cannot arise, and we denote the stationary solution corresponding to $b=0$ and $r_0 = \int_{-\infty}^0 e^{\Theta s}dX_s$ by $U=(U_t)_{t\geq 0}$, given by 
\begin{equation}
\label{eq:U}
    U_t = e^{-\Theta t}\int_{-\infty}^t e^{\Theta s}dX_s,
\end{equation} 
where the integral can be defined via integration-by-parts, see \cite{Marko1}.
The covariance matrix of $U_t$ with time-lag $t$ is denoted by $\gamma(t) = \E[U_tU_0^\top] \in \R^{d\times d}$.
Note that $\gamma(-t) = \gamma(t)^\top$.
For the cross-covariance matrix $\Cov(X_t)$, we use the short notation $V(t)$. 

As in \cite{Marko1}, we use the notations
\begin{align}
    B_t &= \int_0^t \left(\gamma(s)-\gamma(s)^\top \right)ds,\nonumber\\
    C_t &= \int_0^t \int_0^t \gamma(u-s)duds, \label{eq:B-D_def}\\
    D_t &= \sigma V(t)\sigma^\top- \Cov(U_t-U_0)\nonumber.
\end{align}
Note that we can express $\Cov(U_t-U_0)$ with $\gamma$ as $\Cov(U_t-U_0) = 2\gamma(0)-\gamma(t)-\gamma(-t)$.
It was shown in \cite[Theorem 2]{Marko1} that the matrix $\Theta$ satisfies, for any $t\geq 0$, the continuous-time algebraic Riccati equation (CARE)
$$
B_t^\top \Theta + \Theta B_t - \Theta C_t \Theta + D_t = 0. 
$$

In the sequel, we assume that a sample path of the process $r$ is observed continuously on the interval $[0,T]$. Our objective is to estimate the unknown parameters $\Theta$, $b$, and $\sigma$. We assume that the covariance matrix $V(t)$ of the process $X$ 
 is known on $[0,t]$ for some fixed $t$. Note that this is not restrictive, as the uncertainty of the noise is encoded in the unknown $\sigma$.
 
 Estimation of $\sigma$ requires densely spaced discrete observations over a fixed interval. From a theoretical perspective, under a continuous-time observation scheme, one effectively has infinitely many such observations. However, in practice, data are always available only in discrete form. For this reason, we prefer to retain the explicit dependence on the mesh size in our estimators; see Remark 2.6 for further details.

\begin{definition}
We define an estimator for the unknown mean level $b$ as
\begin{equation*}
    \hat b_T =  \frac{1}{T}\int_0^T r_s ds,
\end{equation*}
and an estimator for the unknown covariance matrix $\gamma$ with lag $s$ as
\begin{equation*}
    \hat{\gamma}_{r,T}(s) = \frac{1}{T}\int_0^T r_{s+v}r_v^\top dv - \frac{1}{T^2} \int_0^T r_vdv \int_0^T r_v^\top dv.
\end{equation*}
For $N+1$ equally spaced observations $r_{t_0}, ..., r_{t_{N}}$ on the interval $[0,1]$, we define an estimator 
for $\sigma \sigma^\top$ as
\begin{equation*}
        \hat \sigma \hat \sigma^\top = V(1/N)^{-1} \frac{1}{N}\sum_{k=0}^{N-1} \left(r_{t_{k+1}} - r_{t_k}\right)\left(r_{t_{k+1}} - r_{t_k}\right)^\top,
    \end{equation*}
    provided that the cross-covariance matrix $V(1/N)$ is invertible.
An estimator $\hat\sigma$ for $\sigma$ is then defined as the square root of $\hat \sigma \hat \sigma^\top$.

Finally, an estimator for $\Theta$ is defined as (any) positive definite solution of 
\begin{equation*}
\hat B_t^\top \hat\Theta + \hat\Theta \hat B_t - \hat\Theta \hat C_t \hat\Theta + \hat D_t = 0,
\end{equation*}
provided it exists.
The estimators $\hat B_t, \hat C_t$ and $\hat D_t$ are given by \eqref{eq:B-D_def} with $\gamma(s)$ replaced by $\hat \gamma_{r,T}(s)$ and $\sigma$ replaced by $\hat \sigma$.
\end{definition}

The estimators for $b$ and $\gamma$ depend on the length of the observed time series $T$ (or $T+s$), and the estimates get more accurate if a longer series is observed. Note that the estimators $\hat{B}_t,\hat{C}_t$, and $\hat{D}_t$ depend on a chosen time instance $t$. On a theoretical level, the choice of $t$ does not play a role as long as we know the value $V(t)$. However, in practice, the choice of $t$ affects the performance of the estimators. For small $t$, the matrices $B_t$ and $C_t$ are effectively zero matrices, and this may cause inaccuracies. For large $t$, the estimation of the matrices becomes computationally more expensive and, at the same time, more inaccurate since then one needs to estimate $\gamma(s)$ for larger lags $s$.
The estimator of $\sigma \sigma^\top$ depends on the number of observations on a fixed time interval, and for simplicity, we use the $[0, 1]$ interval here.
To get a better estimate of $\sigma \sigma^\top$, a larger number of observations on this interval is required so that the contribution of the drift term in \eqref{eq:model} vanishes.
Therefore, a good estimate for $\Theta$ requires a long time series and densely observed data.

For obtaining asymptotical results, we make the following assumption.
\begin{ass}
\label{ass:consistency}
$ $
\begin{enumerate}
    \item We have $
    \Vert \gamma(t) \Vert \rightarrow 0$ as $t\rightarrow \infty$.
\item The covariance estimator $\hat \gamma_{U,T}(s)$ based on the stationary process $U$
\begin{equation*}
    \hat{\gamma}_{U,T}(s) = \frac{1}{T}\int_0^T U_{s+v}U_v^T dv
\end{equation*}
satisfies $$\Vert \hat \gamma_{U,T}(\cdot) - \gamma(\cdot)\Vert_{L^\infty(0,t)} \rightarrow 0$$
in probability as $T\rightarrow \infty$, where $\Vert \cdot\Vert_{L^\infty(0,t)}$ denotes the supremum norm on the interval $[0,t]$.
\item For large $N$, the cross-covariance matrix $V(1/N) \in \R^{d\times d}$ is invertible and satisfies 
    \begin{equation*}
        \frac{\Vert V(1/N)^{-1}\Vert}{N^2} = \frac{1}{N^2 \min_i \E[X_{{1/N}, i}^2]} \to 0 \text{ as } N\to \infty.
    \end{equation*}
    \item We have
\begin{equation*}
        V(1/N)^{-1} \frac{1}{N^2} \sum_{k=1}^N \sum_{j=1}^N f(j, k, N) V(1/N)^{-1}\rightarrow 0,
    \end{equation*}
    where 
    \begin{align*}
        f(j, k, N) &= \Cov(\Delta X_{t_k} \Delta X_{t_k}^\top, \Delta X_{t_j} \Delta X_{t_j}^\top) \\
        &= \mathbb{E}\left[\left( \Delta X_{t_k} \Delta X_{t_k}^\top - \E[\Delta X_{t_k} \Delta X_{t_k}^\top] \right) \left(\left( \Delta X_{t_j} \Delta X_{t_j}^\top - \E[\Delta X_{t_j} \Delta X_{t_j}^\top] \right)\right)^\top \right]
    \end{align*}
    with $\Delta X_{t_k} = X_{t_k} - X_{t_{k-1}}$
    and $t_k = \frac{k}{N}$ with $k=0,1,\ldots,N$.
\end{enumerate}
\end{ass}

The first part of Assumption \ref{ass:consistency} (stating that the true covariance of the process $U$ decays to zero), together with Lemma \ref{lemma:mean}, allows us to consistently estimate $b$ from observations of $r$.
This is also required for covariance estimation.
The second part of Assumption \ref{ass:consistency} (uniform consistency of the covariance estimator), together with Lemma \ref{lemma:gamma}, is needed for the estimators of the coefficient matrices $B, C$ and $D$ to be consistent.

The third and the fourth part of Assumption \ref{ass:consistency} are related to the estimation of $\sigma \sigma^\top$.
The third part ensures that we can estimate the quadratic variation of the random process $X$ with the quadratic variation of $r$. 
Note that the third part implies that $\E[X_{1/N,i}^2]$ decays, for all $i=1,\ldots, d$, to zero slower than $N^{-2}$, meaning that none of the components $X_i$ are smooth (in the mean-square sense). 
This implies that the drift term in \eqref{eq:model} is asymptotically negligible in the computation of the quadratic variation. 
Finally, the fourth part ensures that the quadratic variation of $X$ is asymptotically deterministic and behaves as its expectation. 

We now present our main results.
Theorem \ref{thm:consistency} is related to consistency of the estimators and Theorem \ref{thm:CLT} is related to limiting distributions.
\begin{thm}
\label{thm:consistency}
    Let Assumption \ref{ass:consistency} hold. Then, as $T \to \infty$ and $N \to \infty$, we have that
$$
\hat b_T \to b, \ \hat \sigma=\hat \sigma_N \to \sigma \text{ and }\hat\Theta = \hat\Theta_{T, N} \to \Theta
$$
in probability.
\end{thm}
\begin{thm}
\label{thm:CLT}
Let $l_1(T) (\hat b_T - b) \xrightarrow{d} Y_1$, $l_2(T) (\hat \gamma_{U,T}(s) - \gamma(s)) \xrightarrow{d} Y_2(s)$ uniformly in $s\in[0,t]$, and let $l_3(N(T)) (\hat \sigma \hat \sigma^\top - \sigma \sigma^\top) \xrightarrow{d} Y_3$ for some $N(T) \to \infty$ as $T\to \infty$. 

Then, as $T\to \infty$, we have that:
\begin{enumerate}
    \item 
    \begin{equation*}
    l_\Theta(T)\vvec(\hat C_t - C_t, \hat B_t - B_t, \hat D_t - D_t) \xrightarrow{d} L_1(\vvec(Y_1 Y_1^\top, Y_2, Y_3)),
\end{equation*}
where $l_\Theta(T) = \min(l_1^2(T), l_2(T),l_3(N(T))$ and $L_1$ is a linear operator depending on the relations between $l_1^2(T),l_2(T)$, and $l_3(N(T))$.
\item     \begin{equation*}
        l_\Theta(T) \vvec(\hat \Theta - \Theta) \xrightarrow{d} L_2(L_1(\vvec(Y_1 Y_1^\top, Y_2, Y_3))),
    \end{equation*}
    where $L_2$ is a linear operator depending on the parameters $\Theta$, $b$ and $\sigma$, $t$ and the cross-covariance of $X$.
\end{enumerate}

\end{thm}

\begin{remark}
Under continuous observations, the mesh (size of $N$) used in estimating $\sigma$ can, in theory, be chosen freely. 
Thus one can always choose $N=N(T)$ to increase fast enough so that the difference $\hat \sigma \hat \sigma^\top - \sigma \sigma^\top$ does not contribute to the limiting distribution. However, in practice, one has to rely on discrete (dense) observations. The above formulation provides information on the role of the estimation of $\sigma$ in that case. 
\end{remark}
\subsection{Examples and discussion}
Suppose first that $X$ is a Gaussian process. Then $U$ is Gaussian as well. Consider first the case $\Vert\gamma(t)\Vert \sim t^{2H-2}$ for $0<H<1$ as $t\to \infty$. Then Assumption \ref{ass:consistency} parts (i)-(ii) are valid, see \cite{Marko1}. For the rate functions $l_1(T)$ and $l_2(T)$ in Theorem \ref{thm:CLT}, note that
\begin{align*}
    \left\Vert\mathrm{Cov}(\hat b_{U, T})\right\Vert &= \left\Vert\mathrm{Cov}\left(\frac{1}{T} \int_0^T U_s ds \right)\right\Vert \sim \frac{1}{T} \int_1^T s^{2H-2} ds,
\end{align*}
and hence the rate function $l_1(T)$ for $\hat b_{U,T}$ is
\begin{align*}
    l_1(T) = \begin{cases}
        \sqrt{T} &\text{ if } 0<H<\frac{1}{2},\\
        \sqrt{\frac{T}{\log T}} &\text{ if } H=\frac{1}{2},\\
        T^{1-H} &\text{ if } \frac{1}{2} < H < 1.
    \end{cases}
\end{align*}
Similarly,
\begin{align*}
    \left\Vert\mathrm{Cov}(\vvec(\hat \gamma_{U,T}(t))) \right\Vert&=\left\Vert \mathrm{Cov}\left(\vvec\left(\frac{1}{T} \int_0^T U_{t+s} U_t^\top ds\right)\right)\right\Vert \sim \frac{1}{T} \int_1^T s^{4H-4} ds,
\end{align*}
and thus the rate function $l_2(T)$ for $\hat\gamma_{U,T}(t)$ is
\begin{equation*}
    l_2(T) = \begin{cases}
        \sqrt{T} &\text{ if } 0<H<\frac{3}{4}, \\
        \sqrt{\frac{T}{\log T}} &\text{ if } H = \frac{3}{4}, \\
        T^{2-2H} &\text{ if } \frac{3}{4}< H <1.
        \end{cases}
\end{equation*}
Therefore, $l_1^2 > l_2$ for $0<H\leq\frac{3}{4}$ and $l_1^2 = l_2$ for $\frac{3}{4}<H<1$. The case $l_1^2(T)<l_2(T)$ is not possible. 
If the covariance function decays faster than $t^{2H-2}$, e.g., exponentially fast, we obtain the standard rate $l_1(T) = l_2(T) = \sqrt{T}$ for both $\hat{b}_{U,T}$ and $\hat\gamma_{U,T}$. 

For Assumption \ref{ass:consistency} parts (iii)-(iv), note that typically (iii) is automatically satisfied whenever $X$ is not smooth. In particular, (iii) is automatically satisfied whenever the components of $X$ are merely H\"older continuous for some $H\in(0,1)$, see \cite{azmoodeh2014necessary}, and then $\Vert V(t)\Vert \sim t^{2H}$ for some $H\in(0,1)$. For (iv), note that, in the Gaussian case, the matrix $f(j,k,N)$ is a diagonal matrix with diagonal entries of the form $2\eta^2_i+\sum_{v=1,v\neq i}^d \eta_i\eta_v$, where $\eta_v = \mathrm{Cov}(\Delta X^{(v)}_{t_k},\Delta X^{(v)}_{t_j})$. If $X$ is the standard $d$-dimensional fractional Brownian motion (fBm) with Hurst index $H\in(0,1)$, this behaves as $N^{-4H}|j-k|^{4H-4}$. This yields (iv) directly. More generally, in the Gaussian case, the behavior of $\mathrm{Cov}(\Delta X^{(v)}_{t_k},\Delta X^{(v)}_{t_j})$ is closely related to the decay of $\gamma(t)$ determined by the behavior of function $V(t)$ for $t\to \infty$ while $V(1/N)$ is related to the function $V(t)$ for $t\to 0$.

In the non-Gaussian case, Assumption \ref{ass:consistency} (i) and (ii) are also typically valid provided that the memory vanishes as the lags grow without a limit. One can again compute the rates $l_1(T)$ and $l_2(T)$ explicitly from $\left\Vert\mathrm{Cov}(\hat b_{U, T})\right\Vert$ and $ \left\Vert\mathrm{Cov}(\vvec(\hat \gamma_{U,T}(t))) \right\Vert$. However, unlike in the Gaussian case, also $l_1^2(T)<l_2(T)$ is possible, see \cite{lietzen2021modeling} for examples. For Assumption \ref{ass:consistency} parts (iii)-(iv), part (iii) is again related to the smoothness of $X$. For example, under hypercontractivity for the moments, one has $\Vert V(t)\Vert \sim t^{2H}$, exactly as in the Gaussian case, whenever the components are H\"older continuous, see \cite{nummi2024necessary}. Finally, for part (iv), as long as the components of $X$ are independent, the matrix $f(j,k,N)$ is diagonal with $i$:th entry given by
$$
\mathrm{Cov}\left((\Delta X^{(v)}_{t_k})^2,(\Delta X^{(v)}_{t_j})^2\right) + \sum_{v=1,v\neq i}^d \eta_i\eta_v.
$$
Here the asymptotic behavior of $\mathrm{Cov}\left((\Delta X^{(v)}_{t_k})^2,(\Delta X^{(v)}_{t_j})^2\right)$ is linked to $ \left\Vert\mathrm{Cov}(\vvec(\hat \gamma_{U,T}(t))) \right\Vert$ while the asymptotic behavior of the terms $\eta_v$ is linked to the asymptotic $\left\Vert\mathrm{Cov}(\hat b_{U, T})\right\Vert$, and hence any of the terms can be the dominating one. 

As a concrete non-Gaussian example, we mention the $d$-dimensional Hermite process, a non-Gaussian generalization of the fBm. In this case, the behavior of all of the terms is similar to the case of the fBm, and hence all the assumptions are valid. Finally, we mention suitable Levy processes (such as Poisson processes or $\alpha$-stable processes) as an example of non-continuous drivers $X$. In this case, due to the independence of the increments, the covariances decay rapidly, and hence the rates are the standard $\sqrt{T}$. In the presence of jumps, one does not have smooth paths. This provides (iii). Finally, the independence of the increments gives (iv) under suitable assumptions on the Levy measure. The broad range of processes that fulfill the required assumptions allows constructing, e.g., interest rate models that incorporate linear combinations of these processes that enables simultaneous modeling of jumps and various memory effects.

\section{Simulated examples}
\label{sec:sim_ex}
In order to illustrate finite sample behavior and the applicability of the provided modeling approach, we present two different simulated examples. We provide one example where $\Theta$ is a diagonal matrix and one with a non-diagonal $\Theta$.
In both examples, two-dimensional interest rates with fBm as the random process $X$ were simulated using the Euler-Maryama method. The model parameters and the Hurst indices $H$ for the simulations are given in Table \ref{table:sim_par}.
The processes were simulated with 1000 different realizations of fBm, resulting in 1000 sample paths for each set of parameters in both examples.
The length of the time series in every sample was 4000, consisting of 10000 observations.

After simulating, the model parameters $b, \Theta$ and $\sigma$ were estimated from each sample.
The upper bound used in the integrals for estimating the matrices $B, C$ and $D$ was 5.
Figures \ref{fig:0.35_component_diag} and \ref{fig:0.8_component_diag} display, in the case of diagonal $\Theta$, the componentwise differences between the estimates $\hat \Theta$ and the true value $\Theta$ as histograms for $H=0.35$ and $H=0.8$, respectively.
Based on the histograms, it seems that the parameter $\Theta$ can be estimated rather well from the simulated data and, as expected, the differences of each component seem normally distributed.
However, as $H$ increases, the estimation error increases as the estimates for $\sigma$ become more inaccurate, see Figure \ref{fig:sigma_hist}.
The histograms in Figure \ref{fig:0.8_component_diag} also hint to that direction, as the centers of the distributions corresponding to the diagonal elements of $\Theta$ are slightly shifted from zero. This phenomenon is expected as well, since when $H$ increases, one has longer memory inducing slower rate of convergence. 

In the case of non-diagonal $\Theta$ and Brownian motion as the noise, Figure \ref{fig:brownian_component_nondiag} displays the componentwise values of the differences $\hat \Theta-\Theta$ and Figure \ref{fig:brownian_norm_nondiag} displays the Frobenius norms of the differences. The histograms in Figure \ref{fig:brownian_component_nondiag} seem normal. In this case, due to the independence of the increments, one has fast convergence towards normal distribution with mean zero. It follows that then the Frobenius norm is essentially a weighted sum of $\chi^2$-distributed random variables. This is visible in Figure \ref{fig:brownian_norm_nondiag}.

Other parameter combinations yield similar results. Histograms under all our parameter combinations are provided in an electronic supplementary material of the article, see \cite{repo}. Moreover, the repository provides a Python implementation for the proposed approach and the simulation results as numerical tables. 
\renewcommand{\arraystretch}{1.5}
\begin{table}[h]
\caption{Parameters for the simulated examples.}
\begin{tabular}{l|l|l|}
         & Diagonal case & Non-diagonal case \\ \hline
$b$      & $\begin{pmatrix} 0 & 0 \end{pmatrix}$ & $\begin{pmatrix} 1 & 3 \end{pmatrix}$ \\ \hline
$\Theta$ & $\begin{pmatrix} 0.5 & 0 \\ 0 & 0.3 \end{pmatrix}$ & $\begin{pmatrix} 0.5 & 0.1 \\ 0.1 & 0.3 \end{pmatrix}$ \\ \hline
$\sigma$ & $I$ & $\begin{pmatrix} 1 & 0 \\ 0 & 2 \end{pmatrix}$ \\ \hline
$H$ & \{0.35, 0.5, 0.6, 0.8\} &  \{0.35, 0.5, 0.6, 0.8\}                
\end{tabular}
\label{table:sim_par}
\end{table}

\begin{figure}[h!]
\includegraphics[width=0.9\textwidth]{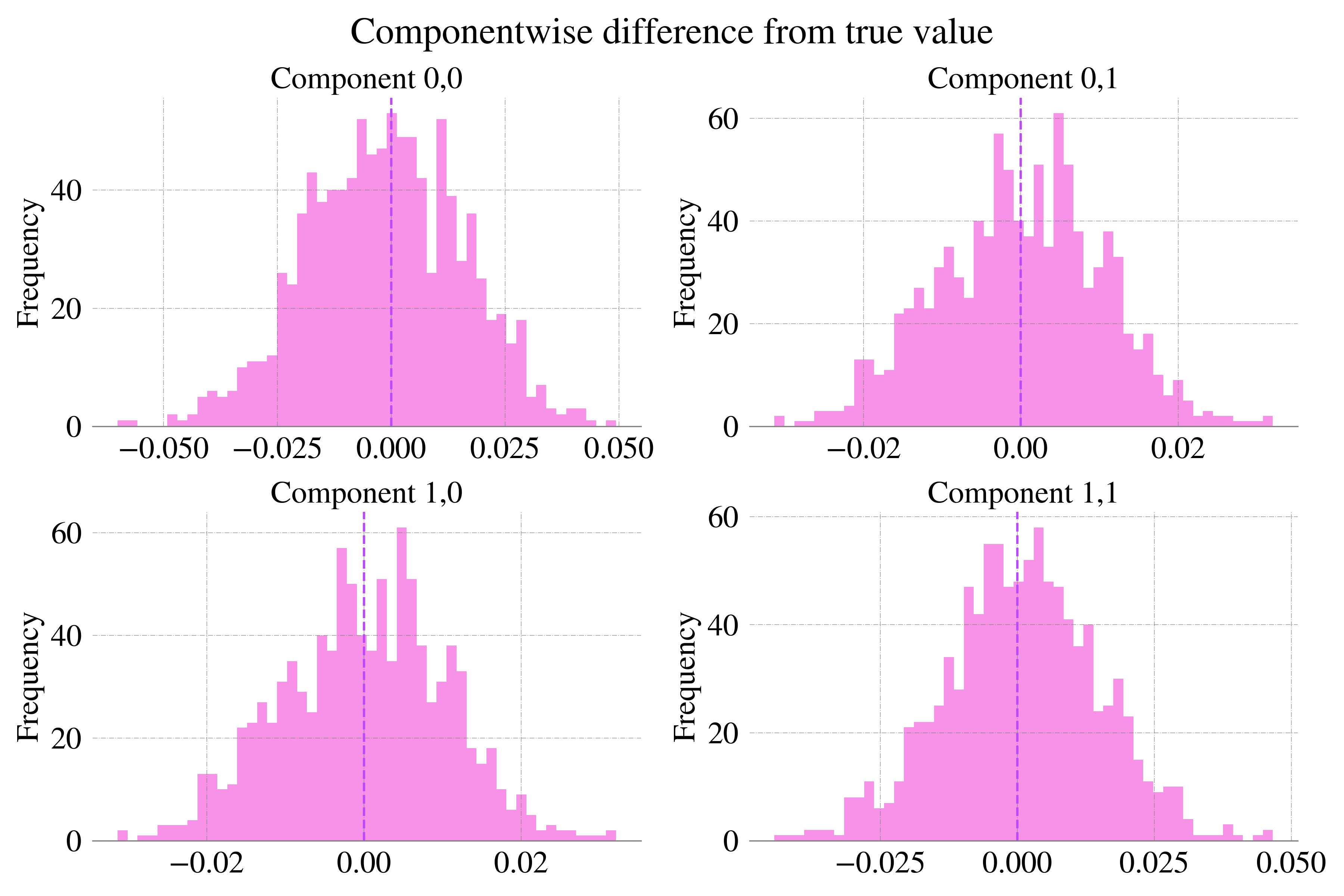}
\caption{Componentwise differences of the estimates $\hat \Theta$ from the true value of $\Theta$ in the diagonal case with $H=0.35$.}
\label{fig:0.35_component_diag}
\end{figure}

\begin{figure}[h!]
\includegraphics[width=0.9\textwidth]{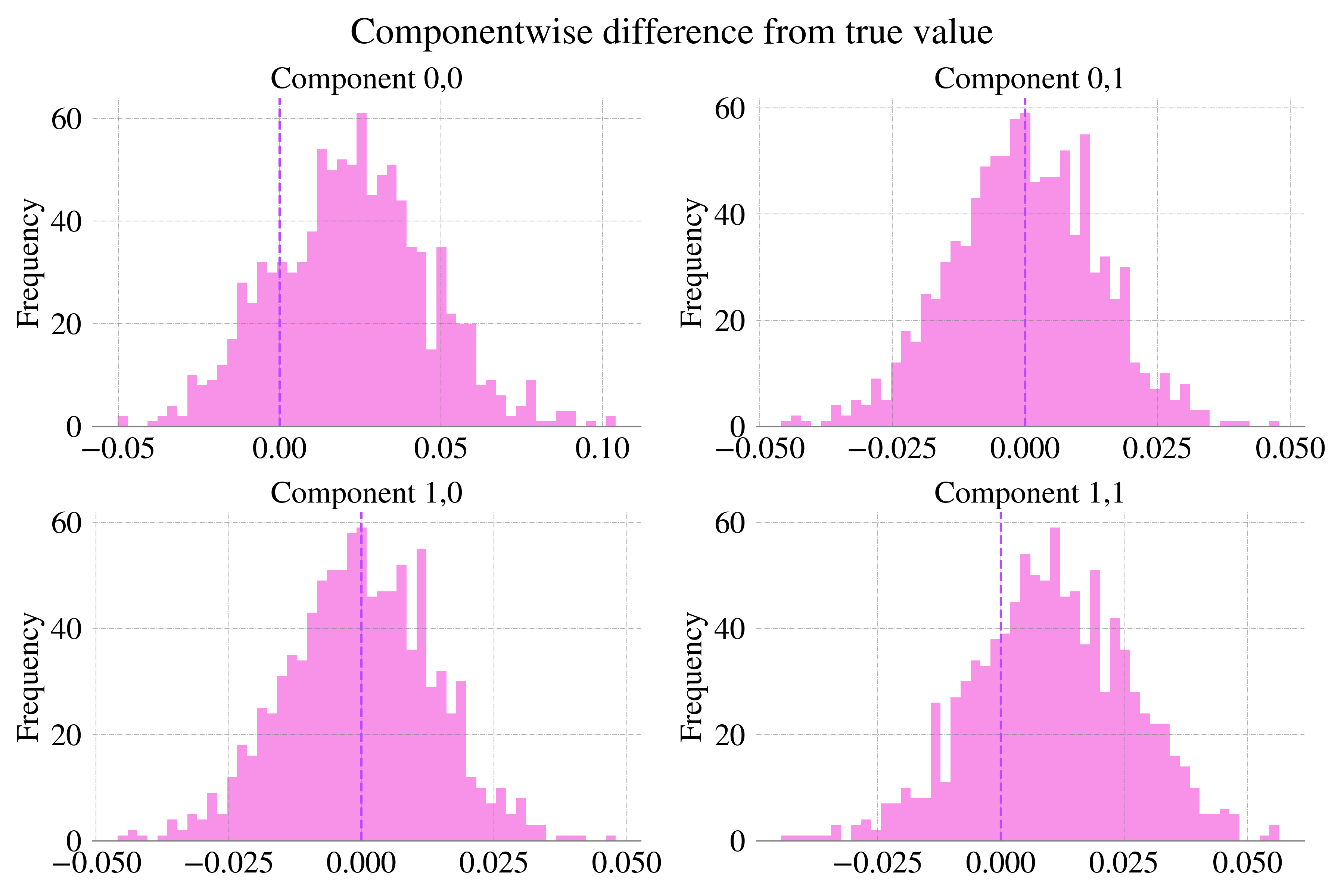}
\caption{Componentwise differences of the estimates $\hat \Theta$ from the true value of $\Theta$ in the diagonal case with $H=0.8$.}
\label{fig:0.8_component_diag}
\end{figure}

\begin{figure}[h!]
     \centering
     \begin{subfigure}{0.9\textwidth}
         \centering
         \includegraphics[width=\textwidth]{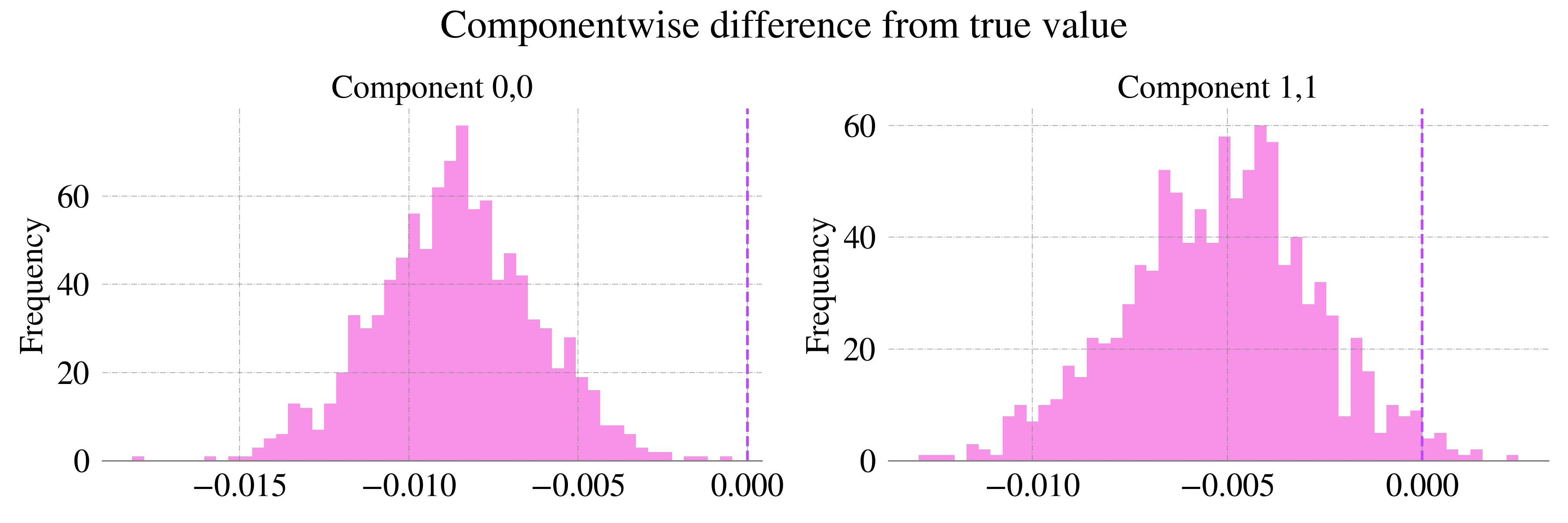}
         \caption{$H=0.35$.}
         \label{fig:sigma_hist_0.35}
     \end{subfigure}
     \hfill
     \begin{subfigure}{0.9\textwidth}
         \centering
         \includegraphics[width=\textwidth]{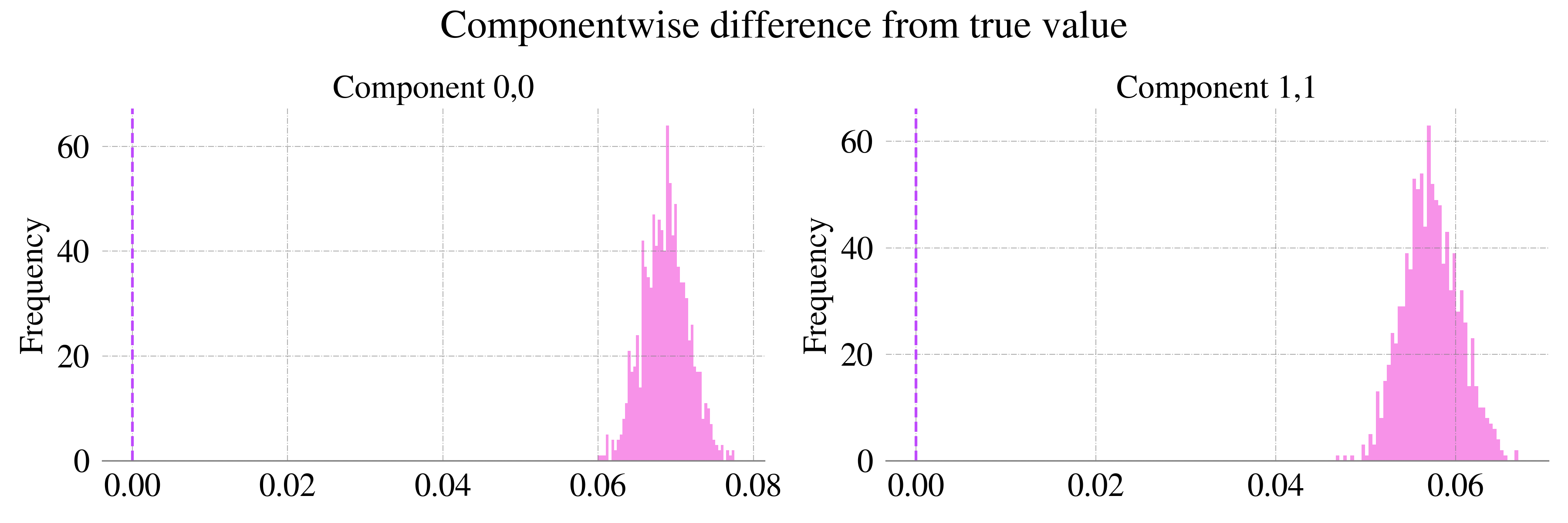}
         \caption{$H=0.8$.}
         \label{fig:sigma_hist_0.8}
     \end{subfigure}
        \caption{Diagonal values of the componentwise differences $\hat \sigma - \sigma$ in the diagonal case with two different Hurst indices.}
        \label{fig:sigma_hist}
\end{figure}

\begin{figure}[h!]
\includegraphics[width=\textwidth]{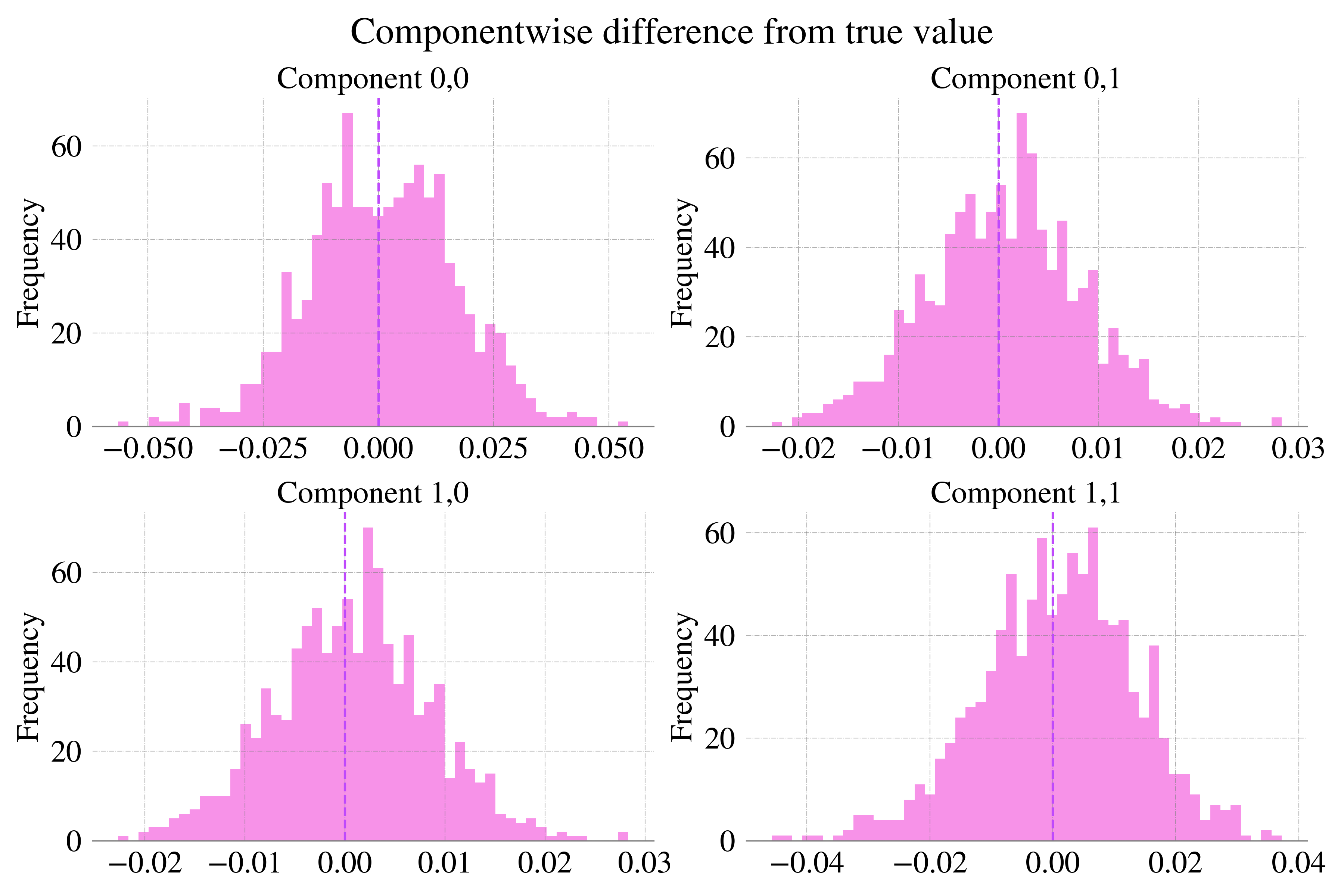}
\caption{Componentwise differences of the estimates $\hat \Theta$ from the true value of $\Theta$ in the non-diagonal case with Brownian noise.}
\label{fig:brownian_component_nondiag}
\end{figure}

\begin{figure}[h!]
\includegraphics[width=0.5\textwidth]{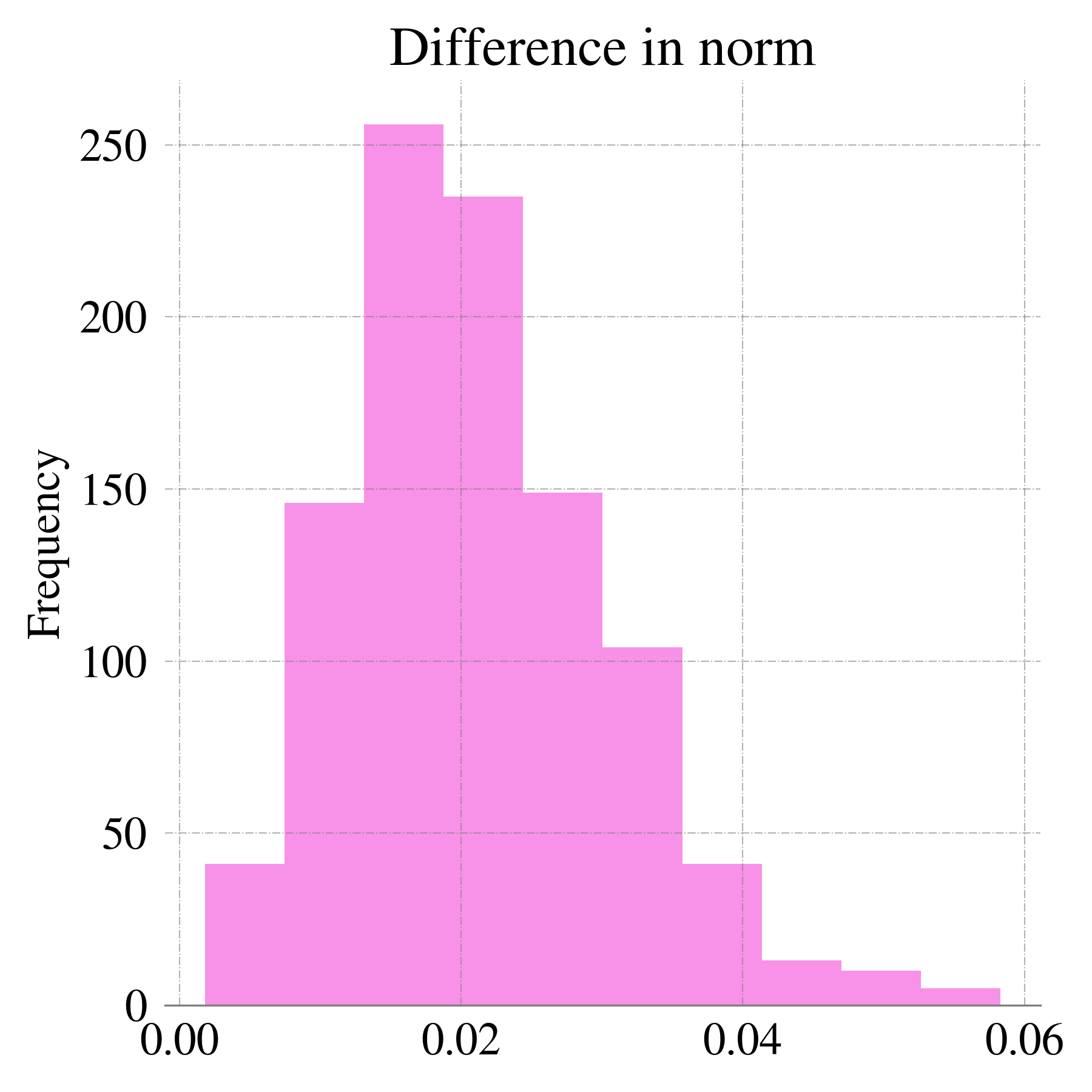}
\caption{Norms of the differences between the estimates $\hat \Theta$ and the true value of $\Theta$ in the non-diagonal case with Brownian noise.}
\label{fig:brownian_norm_nondiag}
\end{figure}

\section{Real data}
\label{sec:real_ex}
To further elaborate the practical usability of our proposed approach, we used it to model bivariate real data. The data consist of 1-month Euribor interest rate \cite{euribor} and the Federal Funds Effective Rate (DFF) \cite{dff} with 6765 daily observations from the time period 4.1.1999 - 6.5.2025.

Based on visualizing the data, the paths of the interest rates are somewhat smooth, indicating that their Hurst indices should be larger than $\frac12$ (see Figure \ref{fig:realdata}). This guided us to use $V(t) \sim t^{2H} I$ for $H\in\left(\frac12,1\right)$. We tried different values within the $\left(\frac12,1\right)$ range, and they all produced similar estimates. Therefore, estimation of the Hurst index does not seem to play a large role here. In the sequel, we report the results for the case $H=0.7$. 

The estimated model parameters were
\begin{align*}
    \hat \Theta &= \begin{pmatrix}
        0.007998 & -0.005324 \\
        -0.005324 & 0.019277
    \end{pmatrix},&\\
    \hat \sigma &= \begin{pmatrix}
        0.02053 & 0 \\
        0  & 0.083129
    \end{pmatrix}, \text{ and} &\\
    \hat b &= \begin{pmatrix}
        1.554624 \\
2.068773
    \end{pmatrix},&
\end{align*}
where the first dimension corresponds to the Euribor rate and the second dimension corresponds to the DFF rate.
Figure \ref{fig:realdata} displays the paths of both interest rates and their respective estimated mean levels $\hat b$.

The components of $\hat \sigma$ can be interpreted as the estimated volatility of the components of $r$.
The value corresponding to DFF is a bit larger than the value corresponding to Euribor. This can be seen when looking 
at the first half of the time series in Figure \ref{fig:realdata}:
the magnitude of fluctuations of DFF is slightly larger than that of the Euribor rate.

The diagonal components of $\hat \Theta$ can be interpreted as the estimated mean reversion rates of the components of $r$.
A larger value pushes the interest rate more strongly towards its mean, whereas a smaller value allows the rate to go and stay further away from the mean.
Therefore, a larger value corresponds to a smaller variance of the interest rate.

The off-diagonal component of $\hat \Theta$ describes the interaction between the components of $r$.
A negative value indicates a positive correlation, meaning that the rates are often on the same side of their means.
Each rate tries to push the other one to the same direction as it itself is: down if the rate itself is below its mean and up if the rate is above its mean.

In the model with Euribor and DFF, the off-diagonal value is negative, and the interest rates seem visibly positively correlated as they are often moving in the same direction and on the same side of their means, see Figure \ref{fig:realdata}.
Between 2016 and 2020, DFF increases while Euribor stays low.
However, even though DFF is increasing during this period, it still stays below the mean level.
Therefore, it is pushing Euribor downwards compensating the mean reversion effect.

A positive off-diagonal value would indicate a negative correlation with similar but opposite interactions between interest rates.
A value close to zero compared to the diagonal values would indicate no correlation between interest rates.

The model assumes stationary increments of the random process $X$. The estimated increments of $X$ for the Euribor-DFF data are visualized in Figure \ref{fig:increment}. Based on the figure, the stationarity assumption is not violated too badly: both increment series seem fat tailed, but they do not show trends or clear heteroskedasticity. 

Finally, to assess the model performance, we provide one-step predictions of the interest rates for the last fifth of the time series (1353 last observations). The predictions are visualized in Figure \ref{fig:prediction}. The predictions seem to be very good.

\begin{figure}[h!]
\includegraphics[width=\textwidth]{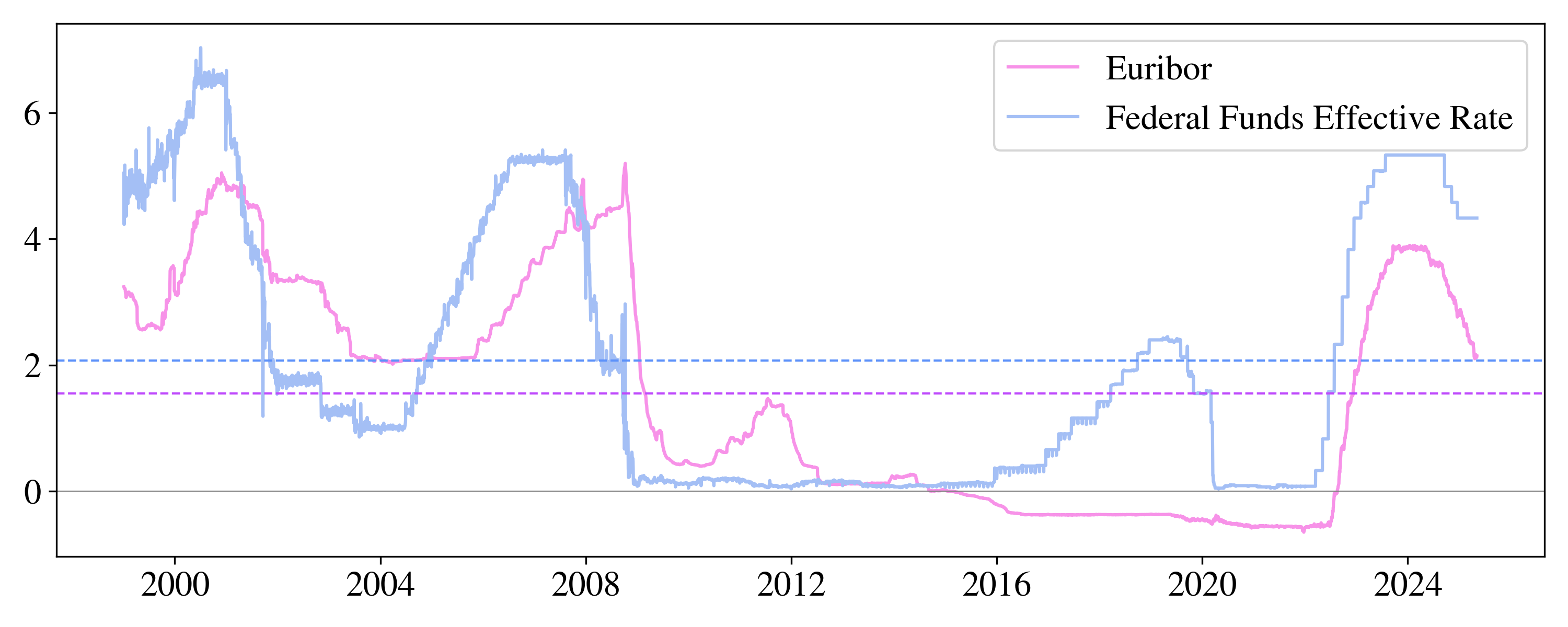}
\caption{Paths of the 1-month Euribor rate and the Federal Funds Effective Rate. The dashed lines indicate the mean levels of both interest rates.}
\label{fig:realdata}
\end{figure}

\begin{figure}[h!]
\includegraphics[width=0.9\textwidth]{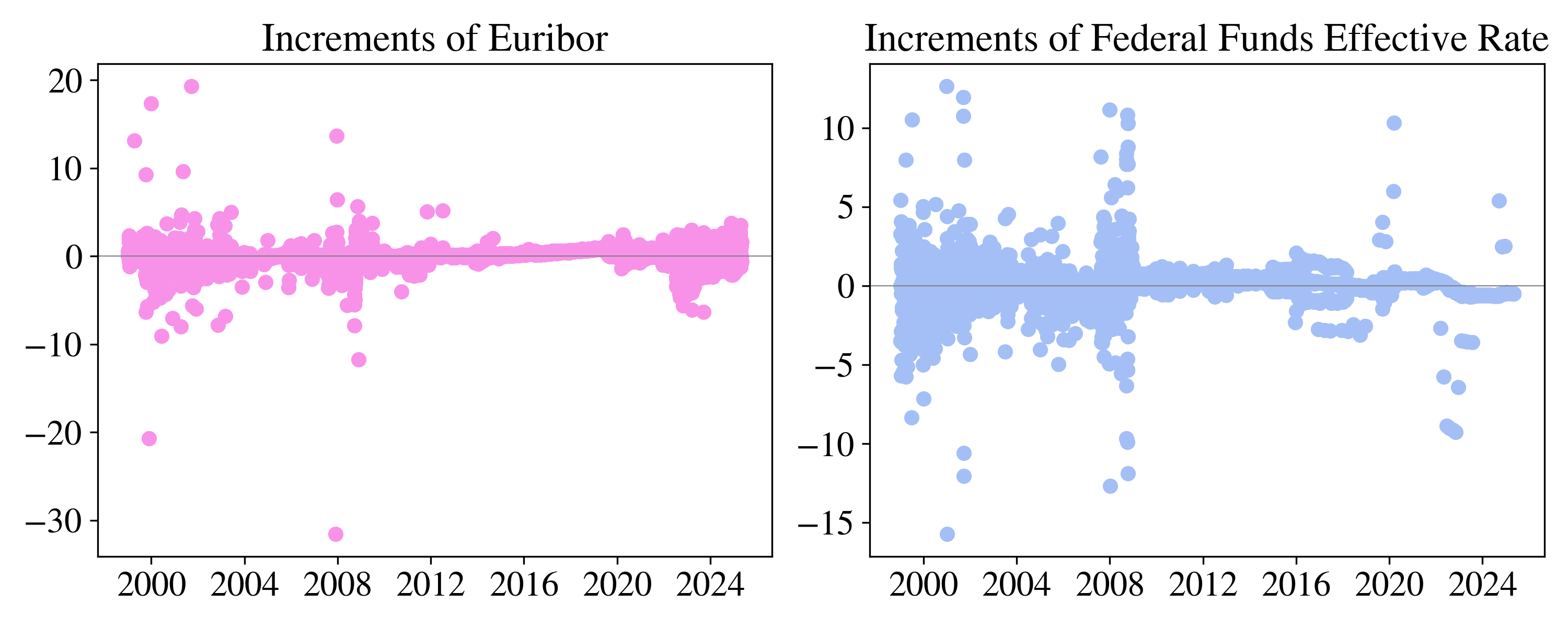}
\caption{Estimated increments of the random process $X$.}
\label{fig:increment}
\end{figure}

\begin{figure}[h!]
\includegraphics[width=0.9\textwidth]{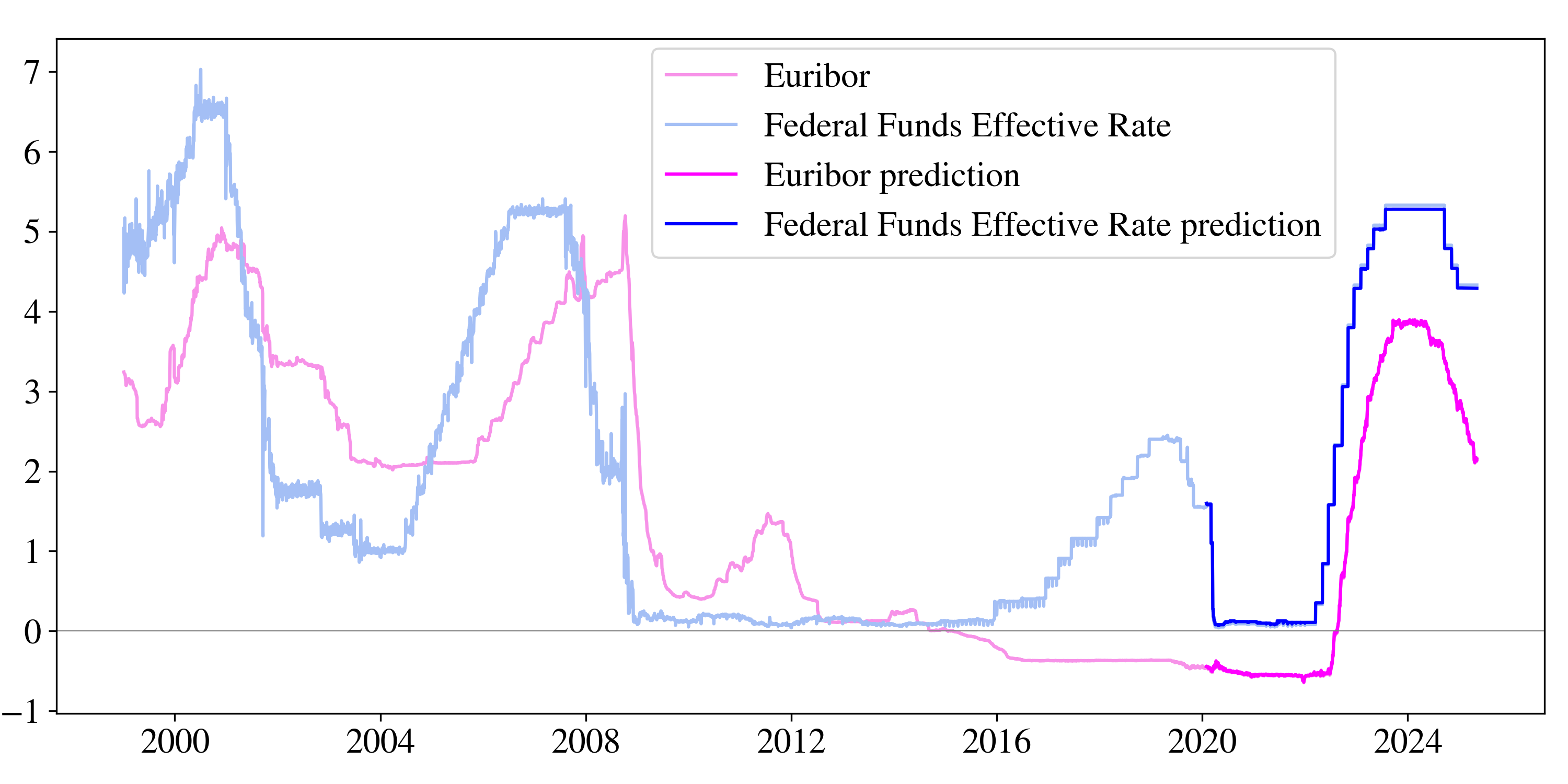}
\caption{Paths of both interest rates and one-step predictions for the last fifth of the time series.}
\label{fig:prediction}
\end{figure}

\section{Proofs}
\label{sec:proof}
All the technical proofs are presented in this section. The section is divided into three subsections.
Section \ref{sec:lemmas} contains the statements and proofs of lemmas that are used in proving the main results.
Section \ref{sec:cons_proof} provides the proof of Theorem \ref{thm:consistency} and Section \ref{sec:clt_proof} contains the proof of Theorem \ref{thm:CLT}.

Throughout this section, we use the $O_P$ notation for stochastic boundedness:
a random object $Z_n$ satisfies $Z_n= O_P(a_n)$ as $n \to \infty$, if for any  $\epsilon >0$ there exist finite $M > 0$ and $N>0$ such that 
\begin{equation*}
    \PR \left(\left\Vert \frac{Z_n}{a_n} \right\Vert > M \right)< \epsilon
\end{equation*}
for all $n>N$. 

\subsection{Auxiliary lemmas}
\label{sec:lemmas}
The first two lemmas in this section, Lemma \ref{lemma:r-U_link} and \ref{lemma:sol}, are fundamental results used throughout the section.
Lemmas \ref{lemma:mean}, \ref{lemma:gamma} and \ref{lemma:sigma} are used in the proof of Theorem \ref{thm:consistency}.
The remaining two lemmas, Lemma \ref{lemma:gamma_lim} and \ref{lemma:d_lim}, are used in the proof of Theorem \ref{thm:CLT}.

\begin{lemma}
\label{lemma:r-U_link}
    The solution $r$ to \eqref{eq:model} and the stationary solution $U$ given by \eqref{eq:U} are linked via 
\begin{equation}
\label{eq:r-U_link}
    r_t = e^{-\Theta t}(r_0-U_0-b)+b + U_t.
\end{equation}
\end{lemma}
\begin{proof}
    By subtracting the stationary equation $dU_t = -\Theta U_t dt + \sigma dX_t$ from \eqref{eq:model}, we have 
    \begin{equation*}
       dr_t - dU_t = \Theta (b-(r_t - U_t))dt. 
    \end{equation*}
Solving this equation for $k_t = r_t - U_t$ with initial value $k_0 = r_0 - U_0$ gives the result.
\end{proof}

\begin{lemma}
\label{lemma:sol}
    The solution to \eqref{eq:model} is given by 
    \begin{equation*}
    r_t = e^{-\Theta t}\left(r_0-b\right)+b + \int_0^t e^{-\Theta(t-s)}dX_s.
\end{equation*}
\end{lemma}
\begin{proof}
    Using 
    \begin{equation*}
        U_t = e^{-\Theta t} \int_{-\infty}^t e^{\Theta s} dX_s =  e^{-\Theta t} U_0 +  \int_0^t e^{-\Theta (t-s)} dX_s
    \end{equation*}
and plugging it into \eqref{eq:r-U_link} yields the claim directly.
\end{proof}

\begin{lemma}
\label{lemma:mean}
We have 
\begin{align*}
&\frac{1}{T}\int_0^T r_v dv - b =\frac{1}{T}\int_0^T U_v dv  + E(T)
\end{align*}
with an error $E(T) = O_P(1/T)$ as $T\to\infty$.
In addition, if $\Vert \gamma(t)\Vert \to 0$ as $t\to \infty$, then, as $T\to\infty$, we have
$$
\frac{1}{T}\int_0^T r_v dv \to b
$$
in probability. 
\end{lemma}
\begin{proof}
    Using Lemma \ref{lemma:r-U_link} and straightforward computations yield
    \begin{align*}
        &\frac{1}{T}\int_0^T r_v dv - b = \frac{1}{T}\int_0^T U_v dv - \Theta^{-1}\frac{I-e^{-\Theta T}}{T}(r_0 - U_0 - b),
    \end{align*}
    where $-\Theta^{-1}\frac{I-e^{-\Theta T}}{T} = \tilde{E}(T)$ is a deterministic error matrix. Since $r_0 - U_0 - b$ is square integrable, it suffices to control $\tilde{E}(T)$. The norm
    \begin{align*}
        \Vert \tilde{E}(T) \Vert &\leq \frac{1}{T} \Vert \Theta^{-1} \Vert \Vert I-e^{-\Theta T}\Vert \leq \frac{1}{T} \Vert \Theta^{-1} \Vert (\Vert I \Vert + \Vert e^{-\Theta T}\Vert).
    \end{align*}
    Since $\Theta$ is positive definite, we can use its eigendecomposition $\Theta = Q\Lambda Q^\top$ (as in \cite{Marko1}) to conclude that $\Vert \Theta^{-1} \Vert = \Vert Q^\top \Lambda^{-1} Q \Vert \leq \Vert \Lambda^{-1} \Vert = \frac{1}{\lambda_{\text{min}}}$, where $\lambda_{\text{min}}$ is the smallest eigenvalue of $\Theta$.
    Similarly, $\Vert e^{-\Theta T}\Vert = \Vert Q e^{-\Lambda T} Q^\top\Vert \leq \Vert e^{-\Lambda T}\Vert = e^{-\lambda_{\text{min}}T}$.
    Using these, we obtain
    \begin{align*}
        \frac{1}{T} \Vert \Theta^{-1} \Vert (\Vert I \Vert + \Vert e^{-\Theta T}\Vert)
        &\leq \frac{1}{T} \frac{1}{\lambda_{\text{min}}} (1 + e^{-\lambda_{\text{min}}T} )
        \leq \frac{1}{T} C,
    \end{align*}
    for a constant $C \geq \frac{2}{\lambda_{\text{min}}}$.
    Therefore, $E(T) = O_P(1/T)$ proving the first part of the claim. 

    For the second part of the claim, the variance of the mean estimator satisfies
    \begin{align*}
        \left\Vert \E[(\frac{1}{T}\int_0^T U_v )dv(\frac{1}{T}\int_0^T U_v dv)^\top] \right\Vert& = \left\Vert \frac{1}{T^2}  \int_0^T\int_0^T \E[U_v U_s^\top] dvds \right\Vert \\
        \leq \frac{1}{T^2}  \int_0^T\int_0^T \left\Vert\gamma(v-s)\right\Vert dvds 
        &=  \frac{2}{T}  \int_0^T\left\Vert\gamma(v)\right\Vert \frac{(T-v)}{T} dv \leq \frac{2}{T}  \int_0^T\left\Vert\gamma(v)\right\Vert dv.
    \end{align*}
    If $\Vert \gamma \Vert$ is integrable, then $\frac{2}{T}  \int_0^T\Vert\gamma(v)\Vert dv \leq \frac{2}{T}  \int_0^\infty \Vert\gamma(v)\Vert dv \rightarrow 0$ as $T \rightarrow \infty$. If $\Vert \gamma \Vert$ is not integrable, L'Hopital's rule gives us
    \begin{align*}
        \lim_{T\to \infty} \frac{\int_0^T\Vert\gamma(v)\Vert dv}{T} = \lim_{T\to \infty} \Vert\gamma(T)\Vert = 0.
    \end{align*}
    Since convergence in $L^2$ implies convergence in probability, this concludes the proof.
\end{proof}

\begin{lemma}
\label{lemma:gamma}

    We have
    \begin{align*}
        \hat{\gamma}_{r,T}(s) &= \frac{1}{T}\int_0^T r_{s+v}r_v^\top dv - \frac{1}{T^2} \int_0^T r_vdv \int_0^T r_v^\top dv \\
        &= \frac{1}{T}\int_0^T U_{v+s}U_v^\top dv + E(T) -(\hat b_T-b)(\hat b_T-b)^\top
    \end{align*}
    with the error $E(T)$ satisfying $E(T) = O_P(1/T)$ as $T\to\infty$.
    In particular, if $\hat{b}_T \to b$ in probability and 
    $$
    \hat\gamma_{U,T}(s)=\frac{1}{T}\int_0^T U_{v+s}U_v^\top dv \to \gamma(s)
    $$
    in probability, uniformly in $s\in[0,t]$, then 
    $$
    \hat\gamma_{r,T}(s) \to \gamma(s)
    $$
    in probability, uniformly in $s\in[0,t]$.
\end{lemma}
\begin{proof}
Plugging in \eqref{eq:r-U_link} to the definition of $\hat{\gamma}_{r,T}(s)$ gives
\begin{align*}
    \hat\gamma_{r,T}(s) = \frac{1}{T}\int_0^T \left[e^{-\Theta (v+s)}(r_0-U_0-b) + b +U_{v+s}\right]\left[e^{-\Theta v}(r_0-U_0-b) + b +U_{v}\right]^\top dv - \hat{b}_T\hat{b}_T^\top.
\end{align*}
Now, we first bound the terms involving the stationary process $U_v$ or $U_{v+s}$. For the term
\begin{equation*}
    \frac{1}{T}\int_0^T e^{-{\Theta(s+v)}} (r_0 - U_0 - b) U_v^\top dv,
\end{equation*}
note that 
\begin{align*}
    \E [\Vert (r_0 - U_0 - b) U_v^\top \Vert] &\leq \E [\Vert r_0 \Vert \Vert U_v \Vert ] + \E [\Vert U_0 \Vert \Vert U_v \Vert] + \E[\Vert b \Vert \Vert U_v \Vert] \\
    &\leq \sqrt{\E[\Vert r_0\Vert^2]\E[\Vert U_v \Vert ^2]} + \sqrt{\E[\Vert U_0 \Vert^2] \E[\Vert U_v \Vert^2]} + \sqrt{\E[\Vert b \Vert^2] \E[\Vert U_v \Vert^2]} \\
    & \leq C.
\end{align*}
Thus, by Markov's inequality, we have that, for any $\epsilon>0$,
\begin{align*}
    \PR \left(\left\Vert \int_0^T e^{-{\Theta(s+v)}} (r_0 - U_0 - b) U_v^\top dv \right\Vert > M \right) 
    &\leq \frac{1}{M} \E[ \int_0^T \Vert e^{-{\Theta(s+v)}} \Vert \Vert (r_0 - U_0 - b) U_v^\top \Vert dv ] \\
    &\leq \frac{1}{M} \int_0^T e^{-\lambda_{\text{min}}(s+v)} \E [\Vert (r_0 - U_0 - b) U_v^\top \Vert] dv \\
    &\leq \frac{C}{M \lambda_{\text{min}}} e^{-\lambda_{\text{min}}s} \\
    &< \epsilon,
\end{align*}
for $M > \frac{C}{\epsilon \lambda_{\text{min}}} e^{-\lambda_{\text{min}}s}$. Hence
\begin{equation*}
    \frac{1}{T}\int_0^T e^{-{\Theta(s+v)}} (r_0 - U_0 - b) U_v^\top dv = O_P\left(\frac{1}{T}\right).
\end{equation*}
The term
\begin{equation*}
    \frac{1}{T}\int_0^T U_{v+s} (r_0 - U_0 - b)^\top e^{-{\Theta v}}dv
\end{equation*}
is bounded by the same argument. 

For the cross terms of $b$ and $U_v$, we can use Lemma \ref{lemma:mean} to get
\begin{align*}
    \frac{1}{T} \int_0^T bU_v^\top dv &= -b(b-\hat b_T)^\top + O_P(1/T)  \text{ and}\\
    \frac{1}{T} \int_0^T U_{s+v}b^\top dv &= -(b-\hat b_T)b^\top + O_P(1/T).
\end{align*}
Since
\begin{align*}
    bb^\top - \hat b_T \hat b_T^\top &= (b+\hat b_T)(b-\hat b_T)^\top + b\hat b_T^\top - \hat b_T b^\top \\
    &= (2b+\hat b_T-b)(b-\hat b_T)^\top + b\hat b_T^\top - \hat b_T b^\top \\
    &= 2b(b-\hat b_T)^\top - (\hat b_T-b)(\hat b_T-b)^\top + b\hat b_T^\top - \hat b_T b^\top,
\end{align*}
we have
\begin{align*}
    \frac{1}{T} \int_0^T bU_v^\top dv + \frac{1}{T} \int_0^T U_{s+v}b^\top dv + bb^\top - \hat b_T \hat b_T^\top \\
    = -(\hat b_T-b)(\hat b_T -b)^\top + O_P(1/T).
\end{align*}
Treating the remaining terms as in Lemma \ref{lemma:mean} completes the proof of the first part of the claim. The second part of the claim follows directly from 
\begin{align*}
    \Vert \hat\gamma_{r,T}(s)-\gamma(s)\Vert \leq \left\Vert \frac{1}{T}\int_0^T U_{v+s}U^\top_vdv - \gamma(s)\right\Vert + \Vert \hat b_T - b\Vert^2 +  O_P(1/T).
\end{align*}
This completes the whole proof.

\end{proof}

\begin{lemma}
\label{lemma:sigma}
   Under Assumption \cref{ass:consistency} parts (iii)-(iv), we have, as $N\to \infty$, that
    \begin{equation*}
        V(1/N)^{-1} \frac{1}{N}\sum_{k=0}^{N-1} \left(r_{t_{k+1}} - r_{t_k}\right)\left(r_{t_{k+1}} - r_{t_k}\right)^\top \rightarrow \sigma \sigma^\top
    \end{equation*}
    in probability.
 \end{lemma}
\begin{proof}
    We start by showing that
    \begin{align*}
        V(1/N)^{-1} \frac{1}{N}\sum_{k=0}^{N-1} \left(r_{t_{k+1}} - r_{t_k}\right)\left(r_{t_{k+1}} - r_{t_k}\right)^\top &\\
        = V(1/N)^{-1}\frac{1}{N} \sigma \sum_{k=0}^{N-1} \left(X_{t_{k+1}} - X_{t_k}\right)\left(X_{t_{k+1}} - X_{t_k}\right)^\top \sigma^\top 
        &+ O_P\left(\frac{1}{N \sqrt{\min_i \E[X_{{1/N}, i}^2]}}\right).
    \end{align*}
For this, notice that
$$
r_{t_{k+1}} - r_{t_k} =  \int_{t_k}^{t_{k+1}} \Theta (b - r_v) dv + \sigma(X_{t_{k+1}} - X_{t_k})
$$
providing 
\begin{align}
     (r_{t_{k+1}} - r_{t_k})(r_{t_{k+1}} - r_{t_k})^\top & = \sigma  \left(X_{t_{k+1}} - X_{t_k}\right)\left(X_{t_{k+1}} - X_{t_k}\right)^\top \sigma^\top  \nonumber\\
     &+ \int_{t_k}^{t_{k+1}} \Theta (b - r_v) dv \left( \int_{t_k}^{t_{k+1}} \Theta (b - r_v) dv \right)^\top \label{eq:sigma_1}\\
     &+\sigma (X_{t_{k+1}} - X_{t_k}) \left(\int_{t_k}^{t_{k+1}} \Theta (b - r_v) dv\right) ^\top \label{eq:sigma_2}\\
     &+ \int_{t_k}^{t_{k+1}} \Theta (b - r_v) dv \left(X_{t_{k+1}} - X_{t_k}\right)^\top \sigma^\top. \label{eq:sigma_3}
\end{align}
In addition, we have
\begin{align*}
     \Vert r_{t_{k+1}} - r_{t_k}\Vert & \leq \sup_{v \in [t_k, t_{k+1}]}\Vert \Theta( b + r_v) \Vert (t_{k+1} - t_k) + \Vert \sigma(X_{t_{k+1}} - X_{t_k}) \Vert,
\end{align*}
where $C_k = \sup_{v \in [t_k, t_{k+1}]}\Vert \Theta( b + r_v) \Vert < \infty$ almost surely.
Since $C_\infty = \sup_{k\in\{0, N-1\}} C_k < \infty$ almost surely, it now follows that
\begin{align*}
    &\left\Vert V(1/N)^{-1} \frac{1}{N}  \sum_{k=0}^{N-1} \int_{t_k}^{t_{k+1}} \Theta (b - r_v) dv \left( \int_{t_k}^{t_{k+1}} \Theta (b - r_v) dv \right)^\top \right\Vert \\
    &\leq \frac{1}{N}\Vert V(1/N)^{-1} \Vert \sum_{k=0}^{N-1} C_k^2 (t_{k+1} - t_k)^2 \\
    &\leq \frac{1}{N}\Vert V(1/N)^{-1} \Vert (\sup_{k\in\{0, N-1\}} C_k)^2 \sum_{k=0}^{N-1} (\frac{k+1}{N} - \frac{1}{N})^2 \\
    &= \frac{\Vert V(1/N)^{-1} \Vert}{N^2} C_\infty^2 =O_P\left(\frac{1}{N^2 \min_i \E[X_{{1/N}, i}^2]}\right)\\
    &= O_P\left(\frac{1}{N \sqrt{\min_i \E[X_{{1/N}, i}^2]}}\right).
\end{align*}
This bound handles the sum related to Term \eqref{eq:sigma_1}. 

Consider next the sum related to Term \eqref{eq:sigma_2}. 
Since $X_t$ has stationary increments, we have
$$\E \left[ \Vert V(1/N)^{-1} \left(X_{t_{k+1}} - X_{t_k}\right) \Vert\right] \leq \sqrt{\sum_{i=1}^d \E[(X_{1/N}- X_0)_i^2]^{-1}} \sim \sqrt{\Vert V(1/N)^{-1} \Vert}.$$
Hence
\begin{align*}
    &\left\Vert V(1/N)^{-1} \frac{1}{N} \sum_{k=0}^{N-1} \sigma (X_{t_{k+1}} - X_{t_k}) \left(\int_{t_k}^{t_{k+1}} \Theta (b - r_v) dv\right) ^\top \right\Vert \\
    &\leq \frac{\Vert \sigma \Vert}{N}\sum_{k=0}^{N-1} \left\Vert V(1/N)^{-1} (X_{t_{k+1}} - X_{t_k}) \left(\int_{t_k}^{t_{k+1}} \Theta (b - r_v) dv\right) ^\top \right\Vert\\
    &\leq \frac{1}{N} \frac{\Vert \sigma \Vert C_\infty}{N} \sum_{k=0}^{N-1} \Vert V(1/N)^{-1} \left(X_{t_{k+1}} - X_{t_k}\right) \Vert\\
    &=O_P\left(\frac{1}{N \sqrt{\min_i \E[X_{{1/N}, i}^2]}}\right).
\end{align*}
The sum related to Term \eqref{eq:sigma_3} can be treated with the same argumentation. Thus, as $\sigma$ and $V(1/N)^{-1}$ commute, it suffices to prove that
 \begin{equation*}
        V(1/N)^{-1} \frac{1}{N}\sum_{k=0}^{N-1} \left(X_{t_{k+1}} - X_{t_k}\right)\left(X_{t_{k+1}} - X_{t_k}\right)^\top \rightarrow I.
    \end{equation*}
For this, note that
    \begin{align}
        &\frac{V(1/N)^{-1}}{N}\sum_{k=0}^{N-1} \left(X_{t_{k+1}} - X_{t_k}\right)\left(X_{t_{k+1}} - X_{t_k}\right)^\top \nonumber\\&=
        \frac{V(1/N)^{-1}}{N}\sum_{k=1}^{N} \left( \Delta X_{t_k} \Delta X_{t_k}^\top - \E\left[\Delta X_{t_k} \Delta X_{t_k}^\top\right] \right) \label{eq:expectation1}\\
        & + \frac{V(1/N)^{-1}}{N}\sum_{k=0}^{N-1} \E\left[\left(X_{t_{k+1}} - X_{t_k}\right)\left(X_{t_{k+1}} - X_{t_k}\right)^\top\right]\label{eq:expectation2}.
    \end{align}
 Here 
 \begin{small}
    \begin{align*}
    &\mathbb{E}\left[\frac{V(1/N)^{-1}}{N}\sum_{k=1}^{N} \left( \Delta X_{t_k} \Delta X_{t_k}^\top - \E[\Delta X_{t_k} \Delta X_{t_k}^\top] \right) \left(\frac{V(1/N)^{-1}}{N}\sum_{k=1}^{N} \left( \Delta X_{t_k} \Delta X_{t_k}^\top - \E[\Delta X_{t_k} \Delta X_{t_k}^\top] \right)\right)^\top \right]\\
    &= \frac{V(1/N)^{-1}}{N^{2}}\sum_{k=1}^N \sum_{j=1}^N f(j,k,N)V(1/N)^{-1}, 
    \end{align*}
    \end{small}
    
\noindent and hence \eqref{eq:expectation1} converges to zero by assumption. To conclude the proof, note that, for \eqref{eq:expectation2}, we have 
    \begin{align*}
        &\frac{V(1/N)^{-1}}{N}\sum_{k=0}^{N-1} \E\left[\left(X_{t_{k+1}} - X_{t_k}\right)\left(X_{t_{k+1}} - X_{t_k}\right)^\top\right] \\
        &= V(1/N)^{-1}\frac{N}{N} \E\left[(\left(X_{t_1} - X_{t_0}\right)\left(X_{t_1} - X_{t_0}\right)^\top\right] 
        = V(1/N)^{-1}\E\left[X_{\frac{1}{N}} X_{\frac{1}{N}} ^\top\right] 
        \sim I.
    \end{align*}
\end{proof}

\begin{lemma}
\label{lemma:gamma_lim}
    If,  for some functions $l_1(T)$ and $l_2(T)$, we have that
    \begin{align*}
        l_1(T) (\hat b_T - b) &\xrightarrow{d} Y_1 \textrm{ and} \\
        l_2(T) (\hat \gamma_{U,T}(s) - \gamma_U(s)) &\xrightarrow{d} Y_2(s) \textrm{ uniformly in } s\in[0,t]
    \end{align*}
    as $T\to\infty$, then
    \begin{equation*}
    \begin{cases}
        l_2(T) (\hat \gamma_{r,T}(s) - \gamma_U(s)) \xrightarrow{d} Y_1 Y_1^\top + Y_2(s), &\text{ if } l_1^2(T) = l_2(T),\\
        l_2(T) (\hat \gamma_{r,T}(s) - \gamma_U(s)) \xrightarrow{d} Y_2(s), &\text{ if } l_1^2(T) > l_2(T),\\
        l_1^2(T) (\hat \gamma_{r,T}(s) - \gamma_U(s)) \xrightarrow{d} Y_1 Y_1^\top, &\text{ if } l_1^2(T) < l_2(T),
    \end{cases}
    \end{equation*}
    uniformly in $s\in[0,t]$.
\end{lemma}
\begin{proof}
    We consider the asymptotic behavior as $T\to \infty$. By Lemma \ref{lemma:gamma}, we have
    \begin{equation*}
        \hat \gamma_{r,T}(t) = \hat \gamma_{U,T}(t) + O_P(1/T) -(\hat b_T-b)(\hat b_T-b)^\top
    \end{equation*}
    and by the Continuous mapping theorem, we have
    \begin{equation*}
        l_1^2(T)(\hat b_T-b)(\hat b_T-b)^\top \xrightarrow{d} Y_1 Y_1^\top.
    \end{equation*}
    If $l_1^2(T) = l_2(T)$, both the first and the third term in $\hat \gamma_{r,T}(t)$ contribute, while the second term vanishes and we get
    \begin{equation*}
        l_2(T) (\hat \gamma_{r,T}(t) - \gamma_U(t)) \xrightarrow{d} Y_2(t) + Y_1 Y_1^\top.
    \end{equation*}
    If $l_1^2(T) > l_2(T)$, we must choose $l_2$ as the rate, since otherwise the first term would diverge.
    Now, $l_2(T)(\hat b_T-b)(\hat b_T-b)^\top \rightarrow 0$ and we are left with
    \begin{equation*}
       l_2(T) (\hat \gamma_{r,T}(t) - \gamma_U(t)) \xrightarrow{d} Y_2(t).
    \end{equation*}
    Finally, if $l_1^2(T) < l_2(T)$, we must choose $l_1^2$ as the rate to keep the third term from diverging.
    Then $l_2^1(T)\hat \gamma_{U,T}(t) \to 0$, leading to
    \begin{equation*}
        l_1^2(T) (\hat \gamma_{r,T}(t) - \gamma_U(t)) \xrightarrow{d} Y_1 Y_1^\top.
    \end{equation*}
\end{proof}

\begin{lemma}
\label{lemma:d_lim}
    If,  for some functions $l_1(T)$, $l_2(T)$, $l_3(N(T))$, and $N=N(T)$ with $N\to \infty$, we have that 
    \begin{align*}
        l_1(T) (\hat b_T - b) &\xrightarrow{d} Y_1, \\
        l_2(T) (\hat \gamma_{U,T}(s) - \gamma_U(s)) &\xrightarrow{d} Y_2(s) \text{ uniformly in }s\in[0,t]\text{, and}\\
        l_3(N(T)) (\hat \sigma \hat \sigma^\top- \sigma \sigma^\top) &\xrightarrow{d} Y_3
    \end{align*}
    as $T\to\infty$, then
    \begin{equation*}
        \begin{cases}
            l_3(N(T))(\hat D_t - D_t) \xrightarrow{d} V(t) Y_3, &\text{ if } l_3(N(T))< \min(l^2_1(T),l_2(T)), \\
            l_3(N(T))(\hat D_t - D_t) \xrightarrow{d} V(t)Y_3 - 2Y_2 (0) + Y_2 (t) + Y_2 (-t), &\text{ if } l_3(N(T)) = l_2(T)\leq l_1^2(T),\\
            l_3(N(T))(\hat D_t - D_t) \xrightarrow{d} V(t)Y_3,  &\text{ if } l_3(N(T)) = l_1^2(T)<l_2(T),\\
            l_2(T)(\hat D_t - D_t) \xrightarrow{d} - 2Y_2 (0) + Y_2 (t) + Y_2 (-t), &\text{ if } l_3(N(T)) > l_2(T),l^2_1(T)\geq l_2(T),\\
            l^2_1(T)(\hat D_t - D_t) \xrightarrow{d} 0, &\text{ if } l_3(N(T)) > l^2_1(T), l^2_1(T)<l_2(T).
        \end{cases}
    \end{equation*}
\end{lemma}

\begin{proof}
    The estimator for $D_t$ is
    \begin{equation*}
        \hat D_t = V(t) \hat \sigma \hat \sigma^\top - 2 \hat \gamma_{r,T}(0) + \hat \gamma_{r,T}(t) + \hat \gamma_{r,T}(-t)
    \end{equation*}
    and therefore the difference 
    \begin{align*}
        \hat D_t - D_t = V(t) (\hat \sigma \hat \sigma^\top - \sigma \sigma^\top) - 2 (\hat \gamma_{r,T}(0) - \gamma_U(0)) \\+ (\hat \gamma_{r,T}(t)- \gamma_U(t)) + (\hat \gamma_{r,T}(-t) - \gamma_U(-t)).
    \end{align*}
    
    We now consider the difference $\hat D_t - D_t$ in the three following cases:
    \begin{enumerate}
        \item $l_3(N(T)) < l_1^2(T)$ and $l_3(N(T)) < l_2(T)$,
        \item $l_3(N(T)) = \min(l_1^2(T), l_2(T))$,
        \item $l_3(N(T)) > \min(l_1^2(T), l_2(T))$.
    \end{enumerate}
    
    In the case (i), we have $l_3(N(T)) < l_1^2(T)$ and $l_3(N(T)) < l_2(T)$, making all but the first term vanish giving us
    \begin{equation*}
        l_3(N(T))(\hat D_t - D_t) \xrightarrow{d} V(t) Y_3.
    \end{equation*}
    
    In the case (ii), we have $l_3(N(T)) = \min(l_1^2(T), l_2(T))$.
    If $l_1^2(T) = l_2(T)$, using the assumptions together with Lemma \ref{lemma:gamma_lim} gives
    \begin{align*}
        l_2(T)(\hat D_t - D_t) \xrightarrow{d} &V(t) Y_3 - 2(Y_2 (0) + Y_1 Y_1^\top)+ (Y_2 (t) + Y_1 Y_1^\top) + (Y_2 (-t) + Y_1 Y_1^\top) \\
        =& V(t) Y_3 - 2Y_2 (0) + Y_2 (t) + Y_2 (-t).
    \end{align*}
    Similarly, if $l_1^2(T) > l_2(T)$, we have
    \begin{equation*}
        l_2(T)(\hat D_t - D_t) \xrightarrow{d} V(t) Y_3 - 2Y_2(0) + Y_2 (t) + Y_2 (-t).
    \end{equation*}
    If $l_1^2(T) < l_2(T)$, we have
    \begin{align*}
        l_1^2(T)(\hat D_t - D_t) \xrightarrow{d} & V(t) Y_3 - 2Y_1 Y_1^\top + Y_1 Y_1^\top + Y_1 Y_1^\top 
        =V(t) Y_3.
    \end{align*}
   
    In the case (iii), we have $l_3(N(T)) > \min(l_1^2(T), l_2(T))$ so the first term in $\hat D_t - D_t$ vanishes as the rate function is chosen to be $\min(l_1^2(T), l_2(T))$.
    Using the assumptions and Lemma \ref{lemma:gamma_lim} similarly as in the case (ii) yields the result.
\end{proof}
\begin{remark}
    Uniform convergence of $\hat \gamma_{U,T}$ is not required for Lemma \ref{lemma:d_lim} as the estimate $\hat D_t$ only uses fixed values of $\hat \gamma_{U,T}$.
    However, as uniform convergence is required for the consistency of other coefficient matrix estimators, we keep it here as well.
\end{remark}
\subsection{Proof of Theorem \ref{thm:consistency}}
\label{sec:cons_proof}
\begin{proof}[Proof of Theorem \ref{thm:consistency}]
Consistency of $\hat b_T$ follows from Lemma \ref{lemma:mean}.
The consistency of $\hat \sigma \hat \sigma^\top$ is the result of Lemma \ref{lemma:sigma}.
As the square root function is continuous, we can use the Continuous mapping theorem to obtain, as $T\to\infty$, that
\begin{equation*}
    \hat \sigma \to \sigma
\end{equation*}
in probability. Lemma \ref{lemma:gamma} ensures that we have $\hat \gamma_{r,T} \to \gamma$ uniformly (in $s\in[0,t]$) in probability as $T \to \infty$ whenever $\hat \gamma_{U,T} \to \gamma$ uniformly. As $V(t)$ was assumed to be known, this shows that the estimators $\hat B_t, \hat C_t$ and $\hat D_t$ are consistent. We can now use \cite[Theorem 3]{Marko1} to get, as $T\to\infty$, that
\begin{equation*}
    \hat \Theta \to \Theta
\end{equation*}
in probability. This concludes the proof.
\end{proof}

\subsection{Proof of Theorem \ref{thm:CLT}}
\label{sec:clt_proof}

\begin{proof}
Some of the elements of our proof are inspired by the proof of \cite[Theorem 4]{Marko1}. 

\textbf{Proof of item (i):} We have
    \begin{align*}
        \Delta_T C = \hat C_t - C_t &= \int_0^t (t-s)(\hat \gamma_{r,T}(s) - \gamma(s)) ds + \int_0^t (t-s) (\hat \gamma_{r,T}(s)^\top - \gamma(s)^\top)ds, \\
        \Delta_T B = \hat B_t - B_t &= \int_0^t (\hat \gamma_{r,T}(s) - \gamma(s))ds + \int_0^t ( \gamma (s)^\top - \hat \gamma_{r,T}(s)^\top)ds, \text{ and} \\
        \Delta_T D = \hat D_t - D_t &= V(t) (\hat \sigma \hat \sigma^\top - \sigma \sigma^\top) -2(\hat \gamma_{r,T}(0) - \gamma(0)) + (\hat \gamma_{r,T}(t) - \gamma(t)) + (\hat \gamma_{r,T}(t)^\top - \gamma(t)^\top).
    \end{align*}
We define a function $L_1$ as follows:
\begin{equation*}
        L_1(\vvec(Y_1 Y_1^\top, Y_2, Y_3)) =
            \begin{pmatrix}
            0 \\
            0 \\
            V(t)Y_3
        \end{pmatrix}, \text{ if } l_3(N(T))< \min(l^2_1(T),l_2(T)),
    \end{equation*}
    \begin{equation*}
        L_1(\vvec(Y_1 Y_1^\top, Y_2, Y_3)) = 
        \begin{pmatrix}
            \int_0^t (t-s) (Y_2 (s) + Y_2 (s)^\top ) ds \\
            \int_0^t (Y_2(s) - Y_2(s)^\top )ds \\
            V(t) Y_3 - 2 Y_2(0) + Y_2(t) + Y_2(-t)
        \end{pmatrix},\text{ if } l_3(N(T)) = l_2(T) < l_1^2(T), 
    \end{equation*}
    \begin{equation*}
        L_1(\vvec(Y_1 Y_1^\top, Y_2, Y_3)) = 
        \begin{pmatrix}
            \int_0^t (t-s) (Y_2 (s) + Y_2 (s)^\top ) ds + t^2 Y_1 Y_1^\top \\
            \int_0^t (Y_2(s) - Y_2(s)^\top )ds \\
            V(t) Y_3 - 2 Y_2(0) + Y_2(t) + Y_2(-t)
        \end{pmatrix},\text{ if } l_3(N(T)) = l_2(T) = l_1^2(T), 
    \end{equation*}
    \begin{equation*}
        L_1(\vvec(Y_1 Y_1^\top, Y_2, Y_3)) = 
        \begin{pmatrix}
            t^2 Y_1 Y_1^\top \\
            0 \\
            V(t)Y_3
        \end{pmatrix},\text{ if } l_3(N(T)) = l_1^2(T) < l_2(T), 
    \end{equation*}
    \begin{equation*}
        L_1(\vvec(Y_1 Y_1^\top, Y_2, Y_3)) = 
        \begin{pmatrix}
            \int_0^t (t-s) (Y_2 (s) + Y_2 (s)^\top ) ds \\
            \int_0^t (Y_2(s) - Y_2(s)^\top )ds \\
            - 2 Y_2(0) + Y_2(t) + Y_2(-t)
        \end{pmatrix},\text{ if } l_3(N(T)) > l_2(T), l_2(T) < l_1^2(T), 
    \end{equation*}
    \begin{equation*}
        L_1(\vvec(Y_1 Y_1^\top, Y_2, Y_3)) = 
        \begin{pmatrix}
            \int_0^t (t-s) (Y_2 (s) + Y_2 (s)^\top ) ds + t^2 Y_1 Y_1^\top\\
            \int_0^t (Y_2(s) - Y_2(s)^\top )ds \\
            - 2 Y_2(0) + Y_2(t) + Y_2(-t)
        \end{pmatrix},\text{ if } l_3(N(T)) > l_2(T)= l_1^2(T), 
    \end{equation*}
    and
    \begin{equation*}
        L_1(\vvec(Y_1 Y_1^\top, Y_2, Y_3)) = 
        \begin{pmatrix}
            t^2 Y_1 Y_1^\top\\
            0\\
            0
        \end{pmatrix},\text{ if } l_3(N(T)) > l_1^2(T), l_1^2(T)<l_2(T).
    \end{equation*}

    The function $L_1$ is linear and continuous. Using Lemma \ref{lemma:gamma_lim}, Lemma \ref{lemma:d_lim}, and the Continuous mapping theorem yield
    \begin{equation*}
        L_1(l_\Theta(T) \vvec(\hat \gamma_{r,T}(s) - \gamma(s)) = l_\Theta(T)\vvec(\Delta_T C, \Delta_T B, \Delta_T D) \xrightarrow{d} L_1(\vvec(Y_1 Y_1^\top, Y_2, Y_3)),
    \end{equation*}
    where $l_\Theta$ denotes the corresponding rate function. 

    \textbf{Proof of item (ii):} We can write $\Theta$ as a function of $B, C$ and $D$ as
    \begin{equation*}
        \Theta = f(\vvec(B, C, D))
    \end{equation*}
    for some continuously differentiable function $f$. 
    Now the Delta method gives
    \begin{align*}
        l_\Theta(T)\vvec(\hat \Theta - \Theta) &= l_\Theta(T)\left( f(\vvec(\hat B, \hat C, \hat D)) - f(\vvec(B, C, D))\right) \\
        &\xrightarrow{d} \nabla f(\vvec(B, C, D))(L_1(\vvec(Y_1 Y_1^\top, Y_2, Y_3)) \\
        &=: L_2((L_1(\vvec(Y_1 Y_1^\top, Y_2, Y_3))).
    \end{align*}
\end{proof}

\bibliographystyle{plain}
\bibliography{main}
\end{document}